\documentclass[a4paper,fleqn]{article}

\usepackage{graphicx}

\usepackage{microtype}

\usepackage{varioref}
\usepackage[british]{refstyle}

\usepackage[utf8]{inputenc}
\usepackage[T1]{fontenc}
\usepackage{cmbright} 
\usepackage{fullpage}

\usepackage{amsmath}
\usepackage{amsthm}
\usepackage{amssymb}
\usepackage{amsfonts}
\usepackage{stmaryrd}

\usepackage{url}
\usepackage{array}
\usepackage{ifthen}
\usepackage{multirow}
\usepackage[colorlinks=true,urlcolor=blue,citecolor=red]{hyperref}

\usepackage{longtable}

\usepackage{comment}
\usepackage{multicol}
\usepackage{colortbl}

\theoremstyle{definition}
\newtheorem{definition}{Definition}
\newtheorem{proposition}{Proposition}
\newtheorem{property}{Property}
\theoremstyle{plain}
\newtheorem{lemma}{Lemma}

\newtheorem{theorem}{Theorem}
\theoremstyle{remark}

\newtheorem{example}{Example}

\def\DEF{\stackrel{\Delta}=}
\def\EQDEF{\stackrel{\Delta}\Leftrightarrow}

\def\f#1{\operatorname{#1}}

\def\card#1{|#1|}
\def\powerset#1{2^{#1}}
\def\N{\mathbb N}

\def\range#1#2{[#1;#2]}

\def\indexes#1{\range 1 {\card{#1}}}
\def\concat{\!::\!}
\def\emptyseq{\varepsilon}

\def\precond#1{ {}^\bullet #1 }
\def\postcond#1{ #1 {}^\bullet}

\def\play{\cdot}

\def\Obj{\mathbf{Obj}}
\def\LS{\mathbf{LS}}

\def\csol{\f{local-paths}}
\def\rcsol#1{\f{local-paths}_{#1}}

\def\precond#1{ {}^\bullet #1 }
\def\postcond#1{ #1 {}^\bullet}

\def\get#1#2{#1(#2)}

\def\anN{\Sigma}
\def\anS{S}
\def\anT{T}
\def\antrl#1#2#3#4{ {#1}_{#2} \xrightarrow{#4} {#1}_{#3}}
\def\antr#1#2#3{{#1}_{#2} \rightarrow {#1}_{#3}}
\def\obj#1#2{{#1}\!\leadsto\!{#2}}
\def\anobj#1#2#3{\obj{{#1}_{#2}}{{#1}_{#3}}}
\def\tr{\f{tr}}
\def\RLCG{\mathcal B}
\def\TRRed{\tr(\RLCG)}
\def\dep{\f{enab}}
\def\orig{\f{orig}}
\def\dest{\f{dest}}
\def\state#1{\langle #1 \rangle}

\def\valid#1{{\mathbf{valid}}_{#1}}

\def\step{\tau}
\def\trace{\pi}
\def\traceb{\varpi}
\def\lpath{\eta}
\def\sinit{{s}}

\def\cb{\f{cb}}
\def\redact{\Psi}

\def\Pint{\texttt{Pint}}

\newcolumntype{:}{>{\global\let\currentrowstyle\relax}}
\newcolumntype{^}{>{\currentrowstyle}}
\newcommand{\rowstyle}[1]{\gdef\currentrowstyle{#1}%
  #1\ignorespaces
}

\usepackage{tikz}
\usetikzlibrary{arrows,shapes}
\usetikzlibrary{positioning} 
\usetikzlibrary{matrix,decorations.pathmorphing,shapes.geometric}
\usetikzlibrary{petri}

\tikzstyle{plot}=[every path/.style={-}]
\tikzstyle{axe}=[gray,->,>=stealth']
\tikzstyle{ticks}=[-,font=\scriptsize,every node/.style={gray}]
\tikzstyle{mean}=[thick,-]
\tikzstyle{interval}=[line width=5pt,red,draw opacity=0.7,-]
\tikzstyle{bounce}=[->,densely dotted,>=stealth']
\tikzstyle{selfhit}=[min distance=5mm,curve to]
\tikzstyle{hitless graph}=[every edge/.style={draw,-}]
\tikzstyle{cross}=[preaction={draw=white,-,line width=6pt,shorten >=15pt}]
\tikzstyle{cross2}=[preaction={draw=white,-,line width=3pt,shorten >=10pt}]
\tikzstyle{grn}=[every node/.style={circle,draw,outer sep=2pt,minimum
size=20pt}]
\tikzstyle{elabel}=[draw=none,sloped,above=-5pt,outer sep=0,font=\scriptsize]
\tikzstyle{inh}=[->,>=|]
\tikzstyle{act}=[->,>=latex]

\tikzstyle{piproc}=[draw,circle,minimum size=22pt]

\tikzstyle{aS}=[every edge/.style={draw,->}]
\tikzstyle{Asol}=[draw,circle,minimum size=5pt,inner sep=0]
\tikzstyle{Aproc}=[draw,minimum width=16pt,minimum height=16pt,inner sep=1pt]
\tikzstyle{Aobj}=[]
\tikzstyle{Ainh}=[-|,shorten >= 1mm]

\definecolor{lightgray}{rgb}{0.8,0.8,0.8}
\definecolor{lightgrey}{rgb}{0.8,0.8,0.8}

\tikzstyle{boxed ph}=[]
\tikzstyle{sort}=[fill=lightgray,rounded corners]
\tikzstyle{process}=[circle,draw,minimum size=15pt,fill=white,
font=\footnotesize,inner sep=1pt]
\tikzstyle{black process}=[process, fill=black,text=white, font=\bfseries]
\tikzstyle{gray process}=[process, draw=black, fill=lightgray]
\tikzstyle{current process}=[process, draw=black, fill=lightgray]
\tikzstyle{process box}=[white,draw=black,rounded corners]
\tikzstyle{tick label}=[font=\footnotesize]
\tikzstyle{tick}=[black,-]
\tikzstyle{hit}=[->,>=angle 45]
\tikzstyle{selfhit}=[min distance=30pt,curve to]
\tikzstyle{bounce}=[densely dotted,>=stealth',->]
\tikzstyle{hl}=[font=\bfseries,very thick]
\tikzstyle{hl2}=[hl]
\tikzstyle{nohl}=[font=\normalfont,thin]

\tikzstyle{local transitions}=[->,>=latex',thick,bend left=30,
				every node/.style={fill=white,inner sep=1pt,outer sep=1pt}]

\newcommand{\currentScope}{}
\newcommand{\currentSort}{}
\newcommand{\currentSortLabel}{}
\newcommand{\currentAlign}{}
\newcommand{\currentSize}{}

\newcounter{la}
\newcommand{\TSetSortLabel}[2]{
  \expandafter\repcommand\expandafter{\csname TUserSort@#1\endcsname}{#2}
}
\newcommand{\TSort}[4]{
  \renewcommand{\currentScope}{#1}
  \renewcommand{\currentSort}{#2}
  \renewcommand{\currentSize}{#3}
  \renewcommand{\currentAlign}{#4}
  \ifcsname TUserSort@\currentSort\endcsname
    \renewcommand{\currentSortLabel}{\csname TUserSort@\currentSort\endcsname}
  \else
    \renewcommand{\currentSortLabel}{\currentSort}
  \fi
  \begin{scope}[shift={\currentScope}]
  \ifthenelse{\equal{\currentAlign}{l}}{
    \filldraw[process box] (-0.5,-0.5) rectangle (0.5,\currentSize-0.5);
    \node[sort] at (-0.2,\currentSize-0.4) {\currentSortLabel};
   }{\ifthenelse{\equal{\currentAlign}{r}}{
     \filldraw[process box] (-0.5,-0.5) rectangle (0.5,\currentSize-0.5);
     \node[sort] at (0.2,\currentSize-0.4) {\currentSortLabel};
   }{
    \filldraw[process box] (-0.5,-0.5) rectangle (\currentSize-0.5,0.5);
    \ifthenelse{\equal{\currentAlign}{t}}{
      \node[sort,anchor=east] at (-0.3,0.2) {\currentSortLabel};
    }{
      \node[sort] at (-0.6,-0.2) {\currentSortLabel};
    }
   }}
  \setcounter{la}{\currentSize}
  \addtocounter{la}{-1}
  \foreach \i in {0,...,\value{la}} {
    \TProc{\i}
  }
  \end{scope}
}

\newcommand{\TTickProc}[2]{ 
  \ifthenelse{\equal{\currentAlign}{l}}{
    \draw[tick] (-0.6,#1) -- (-0.4,#1);
    \node[tick label, anchor=east] at (-0.55,#1) {#2};
   }{\ifthenelse{\equal{\currentAlign}{r}}{
    \draw[tick] (0.6,#1) -- (0.4,#1);
    \node[tick label, anchor=west] at (0.55,#1) {#2};
   }{
    \ifthenelse{\equal{\currentAlign}{t}}{
      \draw[tick] (#1,0.6) -- (#1,0.4);
      \node[tick label, anchor=south] at (#1,0.55) {#2};
    }{
      \draw[tick] (#1,-0.6) -- (#1,-0.4);
      \node[tick label, anchor=north] at (#1,-0.55) {#2};
    }
   }}
}
\newcommand{\TSetTick}[3]{
  \expandafter\repcommand\expandafter{\csname TUserTick@#1_#2\endcsname}{#3}
}

\newcommand{\myProc}[3]{
  \ifcsname TUserTick@\currentSort_#1\endcsname
    \TTickProc{#1}{\csname TUserTick@\currentSort_#1\endcsname}
  \else
    \TTickProc{#1}{#1}
  \fi
  \ifthenelse{\equal{\currentAlign}{l}\or\equal{\currentAlign}{r}}{
    \node[#2] (\currentSort_#1) at (0,#1) {#3};
  }{
    \node[#2] (\currentSort_#1) at (#1,0) {#3};
  }
}
\newcommand{\TSetProcStyle}[2]{
  \expandafter\repcommand\expandafter{\csname TUserProcStyle@#1\endcsname}{#2}
}
\newcommand{\TProc}[1]{
  \ifcsname TUserProcStyle@\currentSort_#1\endcsname
    \myProc{#1}{\csname TUserProcStyle@\currentSort_#1\endcsname}{}
  \else
    \myProc{#1}{process}{}
  \fi
}

\newcommand{\repcommand}[2]{
  \providecommand{#1}{#2}
  \renewcommand{#1}{#2}
}
\newcommand{\THit}[5]{
  \path[hit] (#1) edge[#2] (#3#4);
  \expandafter\repcommand\expandafter{\csname TBounce@#3@#5\endcsname}{#4}
}
\newcommand{\TBounce}[4]{
  (#1\csname TBounce@#1@#3\endcsname) edge[#2] (#3#4)
}

\newcommand{\TState}[1]{
  \foreach \proc in {#1} {
    \node[current process] (\proc) at (\proc.center) {};
  }
}

\title{Goal-Oriented Reduction of Automata Networks}

\author{Lo\"ic Paulev\'e}
\date{%
LRI UMR 8623, Univ. Paris-Sud -- CNRS
\\Universit\'e Paris-Saclay, 91405 Orsay, France
\\
\texttt{loic.pauleve@lri.fr}}

\begin{document}

\maketitle

\begin{abstract}
We consider networks of finite-state machines having local transitions
conditioned by the current state of other automata.
In this paper, we introduce a reduction procedure tailored for
reachability properties of the form ``from global state $\sinit$,
there exists a sequence of transitions leading to a state where an automaton
$g$ is in a local state $\top$''.
By analysing the causality of transitions within the individual automata,
the reduction identifies local transitions which can be removed while preserving
\emph{all} the minimal traces satisfying the reachability property.
The complexity of the procedure is polynomial with the total number of local
transitions, and exponential with the maximal number of local states within an
automaton.
Applied to Boolean and multi-valued networks modelling dynamics of biological systems,
the reduction can shrink down significantly the reachable
state space, enhancing the tractability of the model-checking of large networks.
\end{abstract}

\section{Introduction}
\label{sec:introduction}

Automata networks model dynamical systems resulting from simple interactions between entities.
Each entity is typically represented by an automaton with few internal states which evolve subject
to the state of a narrow range of other entities in the network.
Richness of emerging dynamics arises from several factors including the topology of the interactions, the
presence of feedback loop, and the concurrency of transitions.

Automata networks, which subsume Boolean and multi-valued networks, are
notably used to model dynamics of biological systems, including signalling networks or gene
regulatory networks (e.g.,
\cite{Abou-Jaoude2015,Cohen2015,Grieco2013,Klamt06,tcrsig94,Sahin09,egfr104,WM13-FG}).
The resulting models can then be confronted with biological knowledge, for
instance by checking if some time series data can be reproduced by the
computational model.
In the case of models of signalling or gene regulatory networks, such data
typically refer to the possible activation of a transcription factor, or a gene,
from a particular state of the system, which reflects both the environment and
potential perturbations.
Automata networks have also been used to infer targets to control the behaviour
of the system.
For instance, in \cite{Abou-Jaoude2015,Sahin09}, the author use Boolean networks
to find combinations of signals or combinations of mutations that should alter
the cellular behaviour.

From a formal point of view, numerous biological properties can be expressed in
computation models as reachability properties:
from an initial state, or set of states, the existence of a sequence of transitions
which leads to a desired state, or set of states.
For instance, an initial state can represent a combination of signals/perturbations of a signalling network;
and the desired states the set of states where the concerned transcription factor is active.
One can then verify the (im)possibility of such an activation, possibly by
taking into account mutations, which can be modelled, for instance, as the freezing of some
automata to some fixed states, or by the removal of some transitions.

Due to the increasing precision of biological knowledge, models of networks become larger and larger
and can gather hundreds to thousands of interacting entities making the formal analysis of their
dynamics a challenging task:
the reachability problem in automata networks/bounded Petri nets is
PSPACE-complete \cite{ChengEP95}, which limits its scalability.

Facing a model too large for a raw exhaustive analysis, a natural approach is to reduce its dynamics
while preserving important properties.
Multiple approaches, often complementary, have been explored since decades
to address such a challenge in dynamical and concurrent systems
\cite{Sifakis84,Kurshan94,Bensalem95}.
In the scope of rule-based models of biological networks, efficient static analysis methods have been developed
to lump numerous global states of the systems based on the fragmentation of
interacting components \cite{Feret_IJSI2013};
and to
\emph{a posteriori} compress simulated traces to obtain compact
witnesses of dynamical properties \cite{Danos12-IARCS}.
Reductions preserving the attractors of dynamics (long-term/steady-state
behaviour) have also been proposed for chemical reaction networks \cite{Madelaine14} and Boolean networks
\cite{Naldi11-TCS}.
The latter approach applies to formalisms close to automata networks but does not
preserve reachability properties.
On Petri nets, different structural reductions have been proposed to
reduce the size of the model specification while preserving
bisimulation \cite{Schnoebelen00}, or liveness and LTL properties \cite{Ber86,HP-ppl06}.
Procedures such as the cone of influence reduction \cite{Biere99verifyingsafety}
or relevant subnet computation \cite{Talcott2006} allow to identify
variables/transitions which have no influence on a given dynamical property.
Our work has a motivation similar to the two latter approaches.

\paragraph*{Contribution}
We introduce a reduction of automata networks which identifies
transitions that do not contribute to a given reachability property and hence can be ignored.
The considered automata networks are finite sets of finite-state machines where
transitions between their local states are conditioned by the state of other
automata in the network.
We use a general concurrent semantics where any number of automata can apply one
transition within one step.
We call a \emph{trace} a sequential interleaved execution of steps.

Our reduction preserves all the minimal traces satisfying reachability properties of the form ``from state $s$ there exist
successive steps that lead to a state where a given automaton $g$ is in local state $g_\top$''.
A trace is \emph{minimal} if no step nor transition can be removed
from it and resulting in a sub-trace that satisfies the concerned reachability property.
The complexity of the procedure is polynomial in the number of local transitions,
and exponential in the maximal size of automata.
Therefore, the reduction is scalable for networks of multiple
automata, where each have a few local states.

The identification of the transitions that are not part of any minimal trace is
performed by a static analysis of the causality of transitions within automata.
It extends previous static analysis of reachability properties by
abstract interpretation \cite{PMR12-MSCS,PAK13-CAV}.
In \cite{PMR12-MSCS}, necessary or sufficient conditions for reachability are
derived, but they do not allow to capture all the (minimal) traces towards a
reachability goal.
In \cite{PAK13-CAV}, the static analysis extracts local states, referred to as
cut-sets, which are necessarily reached prior to a given reachability goal.
The results presented here are orthogonal: we identify transitions that are
never part of a minimal trace for the given reachability property.
It allows us to output a reduced model where all such transitions
are removed while preserving all the minimal traces for reachability.
Hence, whereas \cite{PAK13-CAV} focuses on identifying necessary conditions for reachability,
this article focuses on preserving sufficient conditions for reachability.

The effectiveness of our goal-oriented reduction is experimented on actual
models of biological networks and show significant shrinkage of the dynamics of
the automata networks, enhancing the tractability of a concrete verification.
Compared to other model reductions, our goal is similar to
the cone of influence reduction \cite{Biere99verifyingsafety} or
relevant subnet computation \cite{Talcott2006}
mentioned above, which identify variables/transitions that do not impact a given property.
Here, our approach offers a much more fine-grained analysis in
order to identify the sufficient transitions and values of variables that contribute to
the property, which leads to stronger reductions.

\paragraph*{Outline}
\Secref{definitions} sets up the definition and semantics of the automata
networks considered in this paper, together with the local causality analysis
for reachability properties, based on prior work.
\Secref{reduction} first depicts a necessary condition using local causality
analysis for satisfying a reachability property
and then introduce the goal-oriented reduction with the proof of minimal traces
preservation.
\Secref{experiments} shows the efficiency of the reduction on a range of
biological networks.
Finally, \secref{discussion} discusses the results and motivates further work.

\paragraph*{Notations}
Integer ranges are noted $[m;n] \DEF \{m, m+1, \cdots, n\}$.
Given a finite set $A$, $\card A$ is the cardinality of $A$;
$\powerset A$ is the power set of $A$.
Given $n\in \N$,
$x = (x^i)_{i \in [1;n]}$ is a sequence of elements indexed by $i \in [1;n]$;
$\card x = n$;
$x^{m..n}$ is the subsequence $(x^i)_{i\in[m;n]}$;
$x\concat e$ is the sequence $x$ with an additional element $e$ at the end;
$\emptyseq$ is the empty sequence.


\section{Automata Networks and Local Causality}
\label{sec:definitions}


\subsection{Automata Networks}
\seclabel{an}

We declare an Automata Network (AN) with a finite set of finite-state machines
having transitions between their local states conditioned by the state of other
automata in the network.
An AN is defined by a triple $(\anN,\anS,\anT)$ (\defref{cfsm}) where
$\anN$ is the set of automata identifiers;
$\anS$  associates to each automaton a finite set of local states:
if $a\in\anN$, $\anS(a)$ refers to the set of local states of $a$;
and $\anT$  associates to each automaton its local transitions.
Each local state is written of the form $a_i$, where $a\in\anN$ is the automaton in
which the state belongs to, and $i$ is a unique identifier;
therefore given $a_i,a_j\in\anS(a)$, $a_i=a_j$ if and only if $a_i$ and $a_j$
refer to the same local state of the automaton $a$.
For each automaton $a\in\anN$, $\anT(a)$ refers to the set of transitions of
the form $t=\antrl aij\ell$ with $a_i,a_j\in\anS(a)$, $a_i\neq a_j$, and $\ell$
the enabling condition of $t$, formed by a (possibly empty) set of local states
of automata different than $a$ and containing at most one local state of each
automaton.
The \emph{pre-condition} of transition $t$, noted $\precond t$, is the set composed
of $a_i$ and of the local states in $\ell$;
the \emph{post-condition}, noted $\postcond t$ is the set composed of $a_j$
and of the local states in $\ell$.

\begin{definition}[Automata Network $(\Sigma,S,T)$]
\label{def:cfsm}
An \emph{Automata Network} (AN) is defined by a tuple $(\Sigma,S,T)$ where
\begin{itemize}
\item $\Sigma$ is the finite set of automata identifiers;
\item For each $a\in\Sigma$, $S(a) = \{a_i,\dots,a_j\}$ is the finite set of local states of automaton $a$;
$S \DEF \prod_{a\in\Sigma} S(a)$ is the finite set of global states;\\
$\LS \DEF \bigcup_{a\in\Sigma} S(a)$ denotes the set of all the local states.
\item $T = \{ a \mapsto T_a \mid a\in \Sigma \}$, where $\forall a\in\Sigma,
	T_a \subseteq S(a)\times\powerset{\LS\setminus S(a)} \times S(a)$
	with $(a_i,\ell,a_j)\in T_a \Rightarrow a_i\neq a_j$
	and $\forall b\in\Sigma, \card{\ell\cap\anS(b)} \leq 1$,
	is the mapping from automata to their finite set of local transitions.
\end{itemize}

\noindent
We note $a_i\xrightarrow \ell a_j\in T \EQDEF (a_i,\ell,a_j)\in T(a)$
and
$\antr aij\in T \EQDEF \exists\ell\in\powerset{\LS\setminus S(a)}, \antrl
aij\ell\in T$.
Given $t = a_i\xrightarrow\ell a_j\in T$,
$\orig(t)\DEF a_i$,
$\dest(t)\DEF a_j$,
$\dep(t)\DEF \ell$,
$\precond t\DEF \{a_i\}\cup\ell$, and
$\postcond t\DEF \{a_j\}\cup\ell$.
\end{definition}

At any time, each automaton is in one and only one local state, forming the
global state of the network.
Assuming an arbitrary ordering between automata identifiers, the set of global
states of the network is referred to as $\anS$ as a shortcut for $\prod_{a\in\anN}\anS(a)$.
Given a global state $s\in\anS$, $\get sa$ is the local state of automaton $a$
in $s$, i.e., the $a$-th coordinate of $s$.
Moreover we write $a_i\in s \EQDEF \get sa=a_i$; and for any
$ls\in\powerset{\LS}$, $ls\subseteq s\EQDEF \forall a_i\in ls,\get sa=a_i$.

In the scope of this paper, we allow, but do not enforce, the parallel application of transitions in
different automata.
This leads to the definition of a \emph{step} as a set of transitions, with at
most one transition per automaton (\defref{step}).
For notational convenience, we allow empty steps.
The pre-condition (resp. post-condition) of a step $\tau$, noted $\precond \step$
(resp. $\postcond \step$),
extends the similar notions on transitions:
the pre-condition (resp. post-condition) is the union of the pre-conditions
(resp. post-conditions) of composing transitions.
A step $\step$ is \emph{playable} in a state $s\in\anS$ if and only if
$\precond \step \subseteq s$, i.e., all the local states in the pre-conditions of
transitions are in $s$.
If $\step$ is playable in $s$, $s\play\step$ denotes the state after the
applications of all the transitions in $\step$, i.e., where for each transition
$\antrl aij\ell\in\tau$, the local state of automaton $a$ has been replaced with
$a_j$.
\begin{definition}[Step]
\deflabel{step}
Given an AN $(\Sigma,S,T)$, a \emph{step} $\step$ is a subset of local transitions $T$
such that for each automaton $a\in\Sigma$, there is at most one local transition $T(a)$ in $\step$
($\forall a\in\Sigma, \card{(\step\cap T(a))}\leq 1$).

\noindent
We note
$\precond\step\DEF \bigcup_{t\in\step}\precond t$
and
$\postcond\step\DEF \bigcup_{t\in\step}\postcond t \setminus
\{\orig(t)\mid t\in\step\}$.

\noindent
Given a state $s\in S$
where $\step$ is playable ($\precond\step\subseteq s$),
$s\play\step$ denotes the state where
$\forall a\in\anN$,
$\get{(s\play\step)}a=a_j$ if $\exists \antr aij\in\step$,
and
$\get{(s\play\step)}a=\get sa$ otherwise.
\end{definition}

Remark that $\postcond\step \subseteq s\play\step$ and that this definition
implicitly rules out steps composed of incompatible transitions, i.e., where different
local states of a same automaton are in the pre-condition.

A \emph{trace} (\defref{trace}) is a sequence of successively playable steps from a state $s\in\anS$.
The pre-condition $\precond\trace$ of a trace $\trace$ is the set of local
states that are required to be in $s$ for applying $\trace$
($\precond\trace\subseteq s$);
and the post-condition $\postcond\trace$ is the set of local states that are
present in the state after the full application of $\trace$
($\postcond\trace\subseteq s\play\trace$).

\begin{definition}[Trace]
\deflabel{trace}
Given an AN $(\Sigma,S,T)$ and a state $s\in S$,
a \emph{trace} $\pi$ is a sequence of steps such that $\forall i\in\indexes\pi$,
$\precond{\pi^i}\subseteq (s\play\pi^1\play\cdots\pi^{i-1})$.

\noindent
The pre-condition $\precond\trace$ and
the post-condition $\postcond\trace$ are
defined as follows:
for all $n\in\indexes\trace$,
for all $a_i\in\precond\pi^n$,
$a_i\in\precond\pi \EQDEF
\forall m\in\range 1{n-1},
\anS(a)\cap\precond\pi^m = \emptyset$;
similarly,
for all $n\in\indexes\trace$,
for all $a_j\in\postcond{\pi^n}$,
$a_j\in\postcond\trace \EQDEF
\forall m\in\range {n+1}m,
\anS(a)\cap\postcond{\pi^m} = \emptyset$.
If $\trace$ is empty,
$\precond\trace=\postcond\trace=\emptyset$.

\noindent
The set of transitions composing a trace $\trace$ is noted
$\tr(\trace)\DEF\bigcup_{n=1}^{\card \trace}{\trace^n}$.
\end{definition}

Given an automata network $(\anN,\anS,\anT)$ and a state $s\in S$,
the local state $g_\top\in \LS$ is \emph{reachable} from $s$
if and only if
either $g_\top\in s$ or
there exists a trace $\pi$ with
$\precond\pi\subseteq s$ and
$g_\top\in\postcond\pi$.

We consider a trace $\trace$ for $g_\top$ reachability from $s$ is
\emph{minimal}
if and only if there exists no different trace reaching $g_\top$
having each successive step being a subset of a step in $\trace$ with the same
ordering (\defref{minimal}).
Say differently, a trace is minimal for $g_\top$ reachability if no step or
transition can be removed from it without breaking the trace validity or
$g_\top$ reachability.

\begin{definition}[Minimal trace for local state reachability]\deflabel{minimal}
A trace $\pi$ is \emph{minimal} w.r.t. $g_\top$ reachability from $s$
if and only if
there is no trace $\varpi$ from $s$, $\varpi\neq\pi$, $\card\varpi\leq\card\pi$, $g_\top\in\postcond\varpi$,
such that
there exists an injection $\phi: \indexes\varpi\to\indexes\pi$
with $\forall i,j\in\indexes\varpi$,
$i < j \Leftrightarrow \phi(i)<\phi(j)$
and
$\varpi^i \subseteq \pi^{\phi(i)}$.
\end{definition}

Automata networks as presented can be considered as a class of $1$-safe Petri
Nets \cite{BC92} (at most one token per place) having groups of mutually
exclusive places, acting as the automata, and where each transition has one and
only one incoming and out-going arc and any number of read arcs.
The semantics considered in this paper where transitions within different automata
can be applied simultaneously echoes with Petri net step-semantics and
concurrent/maximally concurrent semantics \cite{Janicki86,Priese98,Janicki2015}.
In the Boolean network community, such a semantics is referred to as the
asynchronous generalized update schedule \cite{Aracena09}.



\subsection{Local Causality}
\seclabel{local-causality}

Locally reasoning within one automaton $a$, the reachability of one of its local
state $a_j$ from some global state $s$ with $\get sa=a_i$ can be described by
a (local) \emph{objective},
that we note $\obj{a_i}{a_j}$
(\defref{objective}).

\begin{definition}[Objective]
\label{def:objective}
Given an automata network $(\anN,\anS,\anT$),
an \emph{objective} is a pair of local states $a_i,a_j\in\anS(a)$ of a same automaton
$a\in\anN$ and is denoted $\obj{a_i}{a_j}$.
The set of all objectives is referred to as
$\Obj\DEF\{ \obj{a_i}{a_j} \mid (a_i,a_j)\in \anS(a)\times\anS(a), a\in\anN \}$.
\end{definition}

Given an objective $\anobj aij\in\Obj$, $\csol(\anobj aij)$ is the set of
local acyclic paths of transitions $\anT(a)$ within automaton $a$ from $a_i$ to $a_j$ (\defref{csol}).

\begin{definition}[$\csol$]\deflabel{csol}
Given $\anobj aij\in\Obj$,
if $i=j$, \penalty 0$\csol(\anobj aii)\DEF\{\emptyseq\}$;
if $i\neq j$,
a sequence $\lpath$ of transitions in $T(a)$
is in $\csol(\anobj aij)$ if and only if
$\card\lpath\geq 1$,
$\orig(\lpath^1) = a_i$,
$\dest(\lpath^{\card\lpath})=a_j$,
$\forall n\in\range 1 {\card\lpath-1}$,
$\dest(\lpath^n) = \orig(\lpath^{n+1})$,
and
$\forall n,m\in \indexes\lpath,
n>m\Rightarrow
\dest(\lpath^n)\neq\orig(\lpath^m)$.
\end{definition}

As stated by \ptyref{csol}, any trace reaching $a_j$ from a state containing $a_i$ uses all the transitions
of at least one local acyclic path in $\csol(\anobj aij)$.

\begin{property}
For any trace $\trace$,
for any $a\in\anN$, $a_i,a_j\in\anS(a)$,
$1\leq n\leq m\leq\card\trace$
where
$a_i\in\precond{\trace^n}$
and
$a_j\in\postcond{\trace^m}$,
there exists a local acyclic path $\lpath\in\csol(\anobj aij)$
that is a sub-sequence of $\trace^{n..m}$, i.e.,
there is an injection $\phi:\indexes\lpath \to \range nm$ with
$\forall u,v\in\indexes\lpath,
u<v \Leftrightarrow \phi(u)<\phi(v)$
and
$\lpath^u \in\trace^{\phi(u)}$.
\ptylabel{csol}
\end{property}

A local path is not necessarily a trace, as transitions may be conditioned by
the state of other automata that may need to be reached beforehand.
A local acyclic path being of length at most $\card{\anS(a)}$ with unique
transitions, the number of local acyclic paths is polynomial in the number of transitions
$\anT(a)$ and exponential in the number of local states in $a$.

\begin{example}
Let us consider the automata network $(\anN,\anS,\anT)$, graphically represented
in \figref{example1}, where:
\begin{align*}
\anN & = \{a,b,c,d\}
\\
\anS(a) &= \{a_0, a_1\} &
\anT(a) &= \{ \antrl a01{\{b_0\}}, \antrl a10{\emptyset}\}
\\
\anS(b) &= \{b_0, b_1\} &
\anT(b) &= \{ \antrl b01{\{a_1\}}, \antrl b10{\{a_0\}} \}
\\
\anS(c) &= \{c_0, c_1, c_2\} &
\anT(c) &= \{ \antrl c01{\{a_1\}}, \antrl c10{\{b_1\}},
				\antrl c12{\{b_0\}}, \antrl c02{\{d_1\}}\}
\\
\anS(d) &= \{d_0, d_1\} &
\anT(d) &= \emptyset \phantom{\antrl d00{\{b_0\}}}
\end{align*}
The local paths for the objective $\anobj c02$ are
$\csol(\anobj c02)=\{ \antrl c01{\{a_1\}} \xrightarrow{\{b_0\}} c_2, \antrl c02{\{d_1\}}\}$.
From the state $\state{a_0,b_0,c_0,d_0}$, instances of traces are\\
$
\begin{aligned}
&\{\antrl a01{\{b_0\}}\}\concat
	\{\antrl b01{\{a_1\}},\antrl c01{\{a_1\}}\}\concat
	\{\antrl a10{\emptyset}\}\concat
	\{\antrl b10{\{a_0\}}\}\concat
	\{ \antrl c12{\{b_0\}}\} \enspace;\\
&\{\antrl a01{\{b_0\}}\}\concat\{\antrl c01{\{a_1\}}\}\concat\{ \antrl c12{\{b_0\}}\}\enspace;
\end{aligned}
$\\
the latter only being a minimal trace for $c_2$ reachability.

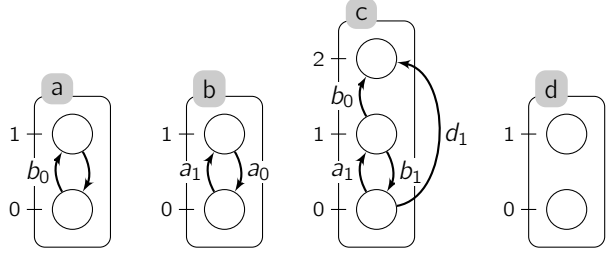
\begin{figure}[t]
\centering
\begin{tikzpicture}
\TSort{(0,0)}{a}{2}{l}
\TSort{(2,0)}{b}{2}{l}
\TSort{(4,0)}{c}{3}{l}
\TSort{(6.5,0)}{d}{2}{l}

\path[local transitions]
  (a_0) edge node[auto] {$b_0$} (a_1)
  (a_1) edge (a_0)
  (c_0) edge node[auto] {$a_1$} (c_1)
  (c_1) edge node[auto] {$b_0$} (c_2)
  (c_1) edge node[auto] {$b_1$} (c_0)
  (c_0) edge[bend right=80] node[right] {$d_1$} (c_2)
  (b_0) edge node[auto] {$a_1$} (b_1)
  (b_1) edge node[auto] {$a_0$} (b_0)
;

\end{tikzpicture}
\caption{An example of automata network.
Automata are represented by labelled boxes, and local states by circles
where ticks are their identifier within the automaton -- for instance, the local
state $a_0$ is the circle ticked 0 in the box $a$.
A transition is a directed edge between two local states within the same
automaton.
It can be labelled with a set of local states of other automata.
In this example, all the transitions are conditioned by at most one other local state.
\label{fig:example1}}
\end{figure}

\end{example}

\section{Goal-Oriented Reduction}
\label{sec:reduction}

Assuming a global AN $(\anN,\anS,\anT)$, an initial state
$\sinit\in\anS$ and a reachability goal $g_\top$ where $g\in\anN$ and
$g_\top\in\anS(g)$, the goal-oriented reduction identifies a subset of local
transitions $\anT$ that are sufficient for producing all the minimal
traces leading to $g_\top$ from $\sinit$.
The reduction procedure takes advantage of the local causality analysis both to
fetch the transitions that matter for the reachability goal and to filter out
objectives that can be statically proven impossible.

\subsection{Necessary condition for local reachability}
\seclabel{gored-filter}

Given an objective $\anobj aij$ and a global state $s\in\anS$ where
$\get sa=a_i$, prior work has demonstrated necessary conditions
for the existence of a trace leading to $a_j$ from $s$ \cite{PMR12-MSCS,PAK13-CAV}.
Those necessary conditions rely on the local causality analysis defined in
previous section for extracting necessary steps that have to be performed in order
to reach the concerned local state.

Several necessary conditions have been established in \cite{PMR12-MSCS}, taking
into account several features captured by the local paths (dependencies,
sequentiality, partial order constraints, \ldots).
The complexity of deciding most of these necessary conditions is polynomial in
the total number of local transitions and exponential in the maximum number of local states
within an automaton.

In this section, we consider a generic reachability over-approximation predicate
$\valid{\sinit}$ which is false only when applied to an objective that
has no trace concretizing it from $\sinit$:
$a_j$ is reachable from $s$ with $\get sa=a_i$ only if $\valid{\sinit}(\anobj aij)$.

\begin{definition}[$\valid{\sinit}$]
Given any objective $\anobj aij\in\Obj$,
$\valid{\sinit}(\anobj aij)$ if there exists a trace $\trace$ from $\sinit$ such that
$\exists m,n\in\indexes\trace$ with
$m\leq n$, $a_i\in\precond{\trace^m}$, and $a_j\in\postcond{\trace^n}$.
\deflabel{valid}
\end{definition}

\def\ValidSet{\Omega}

For the sake of self-consistency, we give in \pporef{oa} an instance
implementation of such a predicate.
It is a simplified version of a necessary condition for reachability demonstrated in \cite{PMR12-MSCS}.
Essentially, the set of valid objectives $\ValidSet$ is built as follows:
initially, it contains all the objectives of the form $\anobj aii$ (that are always valid);
then an objective $\anobj aij$ is added to $\ValidSet$ only if there exists
a local acyclic path $\lpath\in\csol(\anobj aij)$
where all the objectives from the initial state $s$ to the enabling conditions
of the transitions are already in $\ValidSet$:
if $b_k\in\dep(\lpath^n)$ for some $n\in\indexes\lpath$,
then the objective $\anobj b0k$ is already in the set, assuming $\get sb=b_0$.

\def\ite{\f{F}}
\begin{proposition}
\ppolabel{oa}
For all objective $P\in\Obj$,
$\valid{\sinit}(P) \EQDEF P\in\ValidSet$
where
$\ValidSet$ is the least fixed point of the monotonic function
$\ite:\powerset{\Obj}\to\powerset{\Obj}$ with
\begin{align*}
\ite(\Omega) \DEF
\{ \anobj aij\in\Obj&\mid
\exists \lpath\in\csol(\anobj aij): \\
&
\qquad
\forall n\in\indexes\lpath,
\forall b_k\in\dep(\lpath^n),
	\obj{\get{\sinit}b}{b_k}\in\ValidSet \}
\enspace.
\end{align*}
\end{proposition}

Applied to the AN of \figref{example1},
if $\sinit = \langle a_0,b_0,c_0,d_0\rangle$,
$\valid{\sinit}(\anobj c02)$ is true because
$\antrl c01{a_1}\xrightarrow{b_0} c_2\in\csol(\anobj c02)$
with
$\valid{\sinit}(\anobj a01)$ true
and
$\valid{\sinit}(\anobj b00)$ true.
On the other hand,
$\valid{\sinit}(\anobj d01)$ is false.

\medskip
Note that \Pporef{oa} is an instance of $\valid{\sinit}$ implementation; any other implementation satisfying
\defref{valid} can be used to apply the reduction proposed in this article.
In \cite{PMR12-MSCS}, more restrictive over-approximations are proposed.

\subsection{Reduction procedure}
\seclabel{gored-proc}

This section depicts the goal-oriented reduction procedure which
aims at identifying transitions that do not take part in any minimal trace
from the given initial state to the goal local state $g_\top$.
The reduction relies on the local causality analysis
to delimit local paths that may be involved in the goal reachability:
any local transitions that is not captured by this analysis can be removed from the model without
affecting the minimal traces for its occurrence.

The reduction procedure (\defref{rlcg}) consists of collecting a set $\RLCG$ of objectives whose local acyclic paths
may contribute to a minimal trace for the goal reachability.
To ease notations, and without loss of generality, we assume that any
automaton $a$ is in state $a_0$ in $\sinit$.
Given an objective, only the local paths where all the enabling conditions lead to valid objectives
are considered ($\rcsol\sinit$).
The local transitions corresponding to the objectives in $\RLCG$ are noted $\TRRed$.

Initially starting with the main objective $\anobj g0\top$ (\defref{rlcg}(1)), the procedure
iteratively collects objectives that may be involved for the enabling conditions
of local paths of already collected objectives.
If a transition $\antrl bjk\ell$ is in $\TRRed$, for each $a_i\in\ell$, the objective $\anobj a0i$
is added in $\RLCG$ (\defref{rlcg}(2));
and for each other objective $\anobj b\star i\in\RLCG$,
the objective $\anobj bki$ is added in $\RLCG$ (\defref{rlcg}(3)).
Whereas the former criteria references the objectives required for concretizing a local path from the
initial state,
the later criteria accounts for the possible interleaving and successions of
local paths within a same automaton: e.g., $g_\top$ reachability may require to reach
$b_k$ and $b_i$ in some (undefined) order, we then consider 4 objectives: $\anobj b0k$, 
$\anobj bki$,
$\anobj b0i$, 
and $\anobj bik$.

\begin{definition}[$\RLCG$]~
Given an AN $(\anN,\anS,\anT)$,
an initial state $\sinit$ where, without loss of generality, $\forall a\in\anN$, $\get{\sinit}{a}=a_0$,
and a local state $g_\top$ with $g\in\anN$ and $g_\top\in\anS(g)$,
$\RLCG\subseteq\Obj$ is the smallest set which satisfies the following conditions:
\begin{enumerate}
\item $\obj{g_0}{g_\top}\in\RLCG$
\item $\antrl{b}{j}{k}{\ell}\in\TRRed \Rightarrow
	\forall a_i\in\ell, \obj{a_0}{a_i}\in\RLCG$
\item $\antrl{b}{j}{k}{\ell}\in\TRRed
\wedge
	\obj{b_\star}{b_i}\in\RLCG
\Rightarrow
		\obj{b_k}{b_i}\in\RLCG$
\end{enumerate}
\begin{align*}
\text{with}\qquad
\TRRed & \DEF \bigcup_{P\in\RLCG} \tr(\rcsol{\sinit}(P))\enspace,
\text{ where, $\forall P\in\Obj$,}
\\
\rcsol\sinit(P)  & \DEF \{\lpath\in\csol(P)\mid \forall n\in\indexes\lpath,
\\&\qquad\qquad
\forall b_k\in\dep(\lpath^n), \valid{\sinit}(\anobj b0k)\}\enspace,
\label{eq:rcsol}
\end{align*}
$\dep(t)$ being the enabling condition of local transition $t$ (\defref{cfsm}).
\deflabel{rlcg}
\end{definition}

%

\Thmref{gored} states that any trace which is minimal for the reachability of $g_\top$ from initial state
$\sinit$ is composed only of transitions in $\TRRed$.
The proof is given in \supplref{gored-proof}.
It results that the AN $(\anN,\anS,\TRRed)$ contains less transitions
but preserves all the minimal traces for the reachability of the goal.

\begin{theorem}
For each \emph{minimal} trace $\pi$ reaching $g_\top$ from $\sinit$,
$\tr(\pi)\subseteq\TRRed$.
\label{thm:gored}
\end{theorem}

\Figref{example1r} shows the results of the reduction on the example AN
of \figref{example1} for the reachability of $c_2$ from the state where
all automata start at $0$.
Basically, the local path from $c_0$ to $c_2$ using $d_1$ being impossible to
concretize (because $\valid{\sinit}(\anobj d01)$ is false), it has been removed,
and consequently, so are the transitions involving $b_1$ as $b_1$ is not required for
$c_2$ reachability.
In this example, the subnet computation for reachability properties proposed in
\cite{Talcott2006} would have removed only the transition $\antrl c02{d_1}$ from 
\figref{example1}.

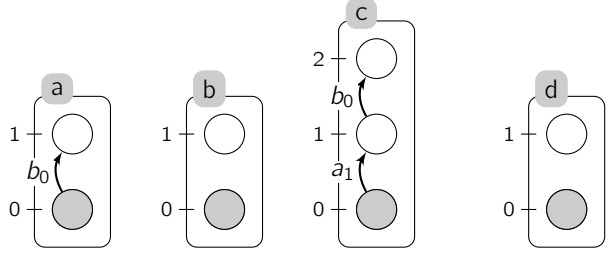
\begin{figure}[t]
\centering
\begin{tikzpicture}
\TSort{(0,0)}{a}{2}{l}
\TSort{(2,0)}{b}{2}{l}
\TSort{(4,0)}{c}{3}{l}
\TSort{(6.5,0)}{d}{2}{l}

\path[local transitions]
  (a_0) edge node[auto] {$b_0$} (a_1)
  (c_0) edge node[auto] {$a_1$} (c_1)
  (c_1) edge node[auto] {$b_0$} (c_2)
;

\TState{a_0,b_0,c_0,d_0}
\end{tikzpicture}
\caption{Reduced automata network from \figref{example1} for the reachability of
$c_2$ from initial state indicated in grey.}
\label{fig:example1r}
\end{figure}

Because the number of objectives is polynomial ($\card{\Obj}=\sum_{a\in\anN} \card{\anS(a)}^2$), the computation
of $\RLCG$ and $\TRRed$ is very efficient, both from a time and space complexity point of view.
The sets $\RLCG\subseteq \Obj$ and $\TRRed\subseteq\anT$ can be built iteratively, from the empty
sets:
when a new objective $\anobj b\star i$ is inserted in $\RLCG$,
each transition in $\tr(\rcsol\sinit(\anobj b\star i))$ is added in $\TRRed$, if not already in;
and for each transition $\antr bjk$ currently in $\TRRed$, the objective $\anobj bki$ is added in
$\RLCG$, if not already in.
When a new transition $\antrl bjk\ell$ is added in $\TRRed$,
for each $a_i\in\ell$, the objective $\anobj a0i$ is added in $\RLCG$, if not already in;
and for each objective $\anobj b\star i$ currently in $\RLCG$, the objective $\anobj bki$ is added in
$\RLCG$, if not already in.

Putting aside the $\tr(\rcsol\sinit)$ computation, the above steps require a polynomial time and a
linear space with respect to the number of transitions and objectives.
The computation of $\tr(\rcsol\sinit(\anobj aij))$ requires a time exponential with the number of local
states in automaton $a$ ($\card{\anS(a)}$), due to the number of acyclic local paths
(\secref{local-causality}), but a quadratic space: indeed, each individual local
acyclic path does not need to be stored, only its set of local transitions, without conditions.
Then, $\valid{\sinit}$ is called at most once per objective.
We assume that the complexity of $\valid{\sinit}$ is polynomial with the number of automata and
transitions and exponential with the maximum number of local states within an automaton
(it is the case of the one presented in \secref{gored-filter})

Overall, the reduction procedure has a polynomial space complexity ($\card{\Obj} + \card{\anT}$) and
time complexity polynomial with the total number of automata
and local transitions, and exponential with the maximum number $k$ of local states within an automaton
($k =\max_{a\in\anN}\card{\anS(a)}$).
Therefore, assuming $k \ll \card{\anN}$, the goal-oriented reduction offers a very low complexity, especially with regard to a full exploration of the $k^{\card{\anN}}$ states.

\section{Experiments}
\label{sec:experiments}

We experimented the goal-oriented reduction on several biological networks and quantify the
shrinkage of the reachable state space.
Then, we illustrate potential applications with the verification of simple reachability, and of cut sets.
In both cases, the reduction drastically increases the tractability of those applications.

\subsection{Results on model reduction}

We conducted experiments on Automata Networks (ANs) that model dynamics of biological networks.
For different initial states, and for different reachability goals, we compared
the number of local transitions in the AN specifications ($\card{\anT}$),
the number of reachable states,
and the size of the so-called complete finite prefix of the unfolding of the net \cite{Esparza08}.
This latter structure is a finite partial order representation of all the possible traces,
which is well studied in concurrency theory.
It aims at offering a compact representations of the reachable state spaces by exploiting the
concurrency between transitions:
if $t_1$ and $t_2$ are playable in a given state and are not in conflict (notably when
$\precond t_1\cap\precond t_2=\emptyset$),
a standard approach would consider 4 global transitions ($t_1$ then $t_2$, and $t_2$ then $t_1$),
whereas  a partial order structure would simply declare $t_1$ and $t_2$ as concurrent, imposing no
ordering between them.
Hence, unfoldings drop part of the combinatorial explosion of the state space due to the interleaving of concurrent transitions.

The selected networks are models of signalling pathways and gene regulatory networks:
two Boolean models of Epidermal Growth Factor receptors (EGF-r) \cite{Sahin09,egfr104},
one Boolean model of tumor cell invasion (Wnt) \cite{Cohen2015},
two Boolean models of T-Cell receptor (TCell-r) \cite{Klamt06,tcrsig94},
one Boolean model of Mitogen-Activated Protein Kinase network (MAPK) \cite{Grieco2013},
one multi-valued model of fate determination in the Vulval Precursor Cells (VPC) in C. elegans \cite{WM13-FG},
one Boolean model of T-Cell differentiation (TCell-d) \cite{Abou-Jaoude2015},
and one Boolean models of cell cycle regulation (RBE2F) \cite{e2frburl}.
The ANs result from automatic translation from the logical network specifications in
the above references; for most models using the \texttt{logicalmodel} tool \cite{logicalmodel}.
Note that the obtained ANs are bisimilar to the logical networks \cite{CHJPS14-CMSB}.
For each of these models, we selected initial states and nodes for which the activation will be the
reachability goal\footnote{Scripts and models
available at \url{http://loicpauleve.name/gored-suppl.zip}}.
Typically, the initial states correspond to various input signal combinations
in the case of signalling cascades, or to pluripotent states for gene networks;
and goals correspond to transcription factors or genes of importance for the model
(output nodes for signalling cascades, key regulators for gene networks).

\Tabref{benchmarks} sums up the results before and after the goal-oriented reduction.
The number of reachable states is computed with \texttt{its-reach} \cite{itstools} 
using a symbolic representation,
and the size of the complete finite prefix (number of instances of transitions) is
computed with \texttt{Mole} \cite{mole}.
The goal-oriented reduction is performed using \Pint{} \cite{PINTurl}.
In each case, the reduction step took less than 0.1s, thanks to its very low complexity when applied to logical networks.

There is a substantial shrinkage of the dynamics for the reduced models, which can turn out
to be drastic for large models.
In some cases, the model is too large to compute the state space without reduction.
For some large models, the unfolding is too large to be computed, whereas it can provide a very
compact representation compared to the state space for large networks exhibiting a high degree of
concurrency (e.g., TCell-d, RBE2F).
In the case of first profile of TCell-d and EGF-r (104)
the reduction removed all the transitions, resulting in an empty model.
Such a behaviour can occur when the local causality analysis statically
detect that the reachability goal is impossible, i.e., the necessary
condition of \secref{gored-filter} is not satisfied.
On the other hand, a non-empty reduced model does not guarantee the goal reachability.
\Supplref{partial} show additional results with the reduction made without the filtering
$\valid{\sinit}$ (\secref{gored-filter}).

\def\mcKO{\multicolumn{2}{c|}{KO}}
\def\mcKOb{\multicolumn{2}{c|}{\bf KO}}
\setlength{\tabcolsep}{1.5mm}
 \renewcommand{\arraystretch}{1.1}
\begin{table}[p]
\centering
\begin{tabular}{|:l||^r|^r|^r||^r^l|^r^l|}
\cline{5-8}
\multicolumn{4}{c|}{} & \multicolumn{4}{c|}{Verification of goal reachability}\\
\hline
Model & $\card{\anT}$ & \# states & |unf| & 
\multicolumn{2}{c|}{\texttt{NuSMV}} & \multicolumn{2}{c|}{\texttt{its-reach}}
\\\hline
\multirow{2}{*}{EGF-r (20)}
& 68 & 4,200 & 1,749 & 0.2s & 10Mb & 0.17 & 7Mb \\
& \textbf{43} & \textbf{722} & \bfseries 336
 & \textbf{0.1} & \textbf{8Mb} & \textbf{0.1s} & \textbf{5Mb}\\\hline
\multirow{2}{*}{Wnt (32)}
	& 197 & 7,260,160 & KO & 30s & 48Mb & 0.3s & 18Mb\\
\rowstyle{\bf}
	& 117 & 241,060 & 217,850 & 0.9s & 32Mb & 0.5s & 17Mb \\\hline
\multirow{2}{*}{TCell-r (40)}
	& 90 & $\approx 1.2\cdot 10^{11}$ & KO & \mcKO & 1.1s & 52Mb\\
\rowstyle{\bf}
	& 46 & 25,092 & 14,071 & 3.8s & 36Mb & 0.6s & 15Mb \\\hline
MAPK (53)
	& 173 & $\approx 3.8\cdot 10^{12}$ & KO & \mcKO & 0.9s & 60Mb\\
profile 1 \rowstyle{\bf}
	& 113 & $\mathbf{\approx 4.5\cdot 10^{10}}$ & KO & \mcKOb & 2s & 48Mb\\\hline
MAPK (53)
	& 173 & 8,126,465 & KO & 63s & 83Mb & 0.2s & 15Mb\\
profile 2 \rowstyle{\bf}
	& 69 & 269,825 & 155,327 & 1.5s & 36Mb & 0.4s & 18Mb \\\hline
\multirow{2}{*}{VPC (88)}
	& 332 & KO & KO & \mcKO & 1s & 50Mb\\
\rowstyle{\bf}
	& 219 & $\mathbf{1.8\cdot 10^9}$ & 43,302 & 236s & 156Mb & 0.8s & 21Mb\\\hline
\multirow{2}{*}{TCell-r (94)}
 &	217& KO  & KO & \multicolumn{2}{c|}{KO} &\multicolumn{2}{c|}{KO} \\
& \textbf{42} & \textbf{54.921} & \bf 1,017 & \textbf{0.4} & \textbf{23Mb} & \textbf{0.26s} & \textbf{14Mb}
\\\hline
TCell-d (101) &
	384 & $\approx 2.7\cdot 10^8$ & 257 & 3s & 40Mb & 0.5s & 24Mb \\
profile 1 \rowstyle{\bf} &
	0 & 1 & 1 & \multicolumn{4}{c|}{\cellcolor{gray}}\\\hline
TCell-d (101) &
	384 & KO & KO & \mcKO & 0.5s & 23Mb \\
profile 2 \rowstyle{\bf} &
	161 & 75,947,684 & KO & 474s & 260Mb & 0.3s & 19Mb\\\hline
EGF-r (104) & 
	378 & 9,437,184 & 47,425 & 7s & 35Mb & 0.6s & 23Mb \\
profile 1 & 
	\textbf{0} & \textbf{1} & \textbf{1} & \multicolumn{4}{c|}{\cellcolor{gray}}
\\\hline
EGF-r (104) &
	378 & $\approx 2.7 \cdot 10^{16}$ & KO & \multicolumn{2}{c|}{KO}& 1.36s & 60Mb
\\
profile 2 \rowstyle{\bf} & 
\textbf{69} & 62,914,560 &  KO & \textbf{11s} & \textbf{33Mb} & \textbf{0.3s} & \textbf{17Mb}
\\\hline
\multirow{2}{*}{RBE2F (370)} &
	742 & KO & KO & \mcKO &\mcKO\\
\rowstyle{\bf} &
	56 & 2,350,494 & 28,856 & 5s & 377Mb & 5s & 170Mb \\\hline
\end{tabular}

\caption{
Comparisons before (normal font) and \textbf{after} (bold font) the goal-oriented AN reduction.
Each model is identified by the system, the number of automata (within
parentheses), and a profile specifying the initial state and the reachability goal.
$\card{\anT}$ is the number of local transitions in the AN specification;
``\#states'' is the number of reachable global states from the initial state;
``|unf|'' is the size of the complete finite prefix of the unfolding.
``KO'' indicates an execution running out of time (30 minutes) or memory.
When applied to goal reachability, we show the total execution time and memory used by the tools
\texttt{NuSMV} and \texttt{its-reach}.
Computation times where obtained on an Intel\textregistered{}
Core\texttrademark{} i7 3.4GHz CPU with 16GB RAM.
\emph{For each case, the reduction procedure took less than 0.1s}.
\tablabel{benchmarks}}
\end{table}
\begin{table}[p]
\centering
\begin{tabular}{|:l|^c|^c|^c|^c|^c|}
\cline{2-6}
\multicolumn{1}{c|}{}
\rowstyle{\normalfont}
& Wnt (32) & TCell-r (40) & EGF-r (104) & TCell-d (101)& RBE2F (370) \\\hline
\multirow{2}{*}{\texttt{NuSMV}} &
	44s 55Mb & KO & KO & KO & KO \\
\rowstyle{\bf} &
	9.1s 27Mb & 2.4s 34Mb & 13s 33Mb & 600s 360Mb & 6s 29Mb\\\hline
\multirow{2}{*}{\texttt{its-ctl}} &
	105s 2.1Gb & 492s 10Gb & KO & KO & KO\\
\rowstyle{\bf} &
	16s 720Mb & 11s 319Mb & 21s 875Mb & KO & 179s 1.8Gb \\\hline
\end{tabular}
\caption{
Comparisons before (normal font) and \textbf{after} (bold font) the goal-oriented AN reduction
for CTL model-checking of cut sets.}
\label{tab:cutsets}
\end{table}

\subsection{Example of application: goal reachability}

In order to illustrate practical applications of the goal-oriented model reduction, we
first
systematically applied model-checking for the goal reachability on the initial and reduced model
(\tabref{benchmarks}).

We compared two different softwares:
\texttt{NuSMV} \cite{NuSMV2} which combines Binary Decision Diagrams and SAT approaches for synchronous systems,
and \texttt{its-reach} \cite{itstools} which implements efficient decision diagram data structures
\cite{Hamez09}.
In both cases, the transition systems specified as input of these tools is an
exact encoding of the asynchronous semantics of the automata networks, where steps
(\defref{step}) are always composed of only one transition.
For \texttt{NuSMV}, the reachability property is specified with CTL \cite{ClarkeEmerson81}
(``\texttt{EF} ${g_\top}$'', $g_\top$ being the goal local state, and \verb|EF| the \emph{exists eventually}
CTL operator).
It is worth noting that \texttt{NuSMV} implements the \emph{cone of influence} reduction
\cite{Biere99verifyingsafety} which removes variables not involved in the property.
\texttt{its-reach} is optimized for checking if a state belongs to the reachable state
space, and cannot perform CTL checking.

Experiments show a remarkable gain in tractability for the model-checking of
reduced networks.
For large cases, we observe that the dynamics can be tractable only after model reduction (e.g.,
TCell-r (94), RBE2F (370)).
\texttt{its-reach} is significantly more efficient than \texttt{NuSMV}
because it is tailored for simple reachability checking, whereas \texttt{NuSMV}
handles much more general properties.

Because the goal-reduction preserves all the minimal traces for the goal reachability, it preserves
the goal reachability: the results of the model-checking is equivalent in the initial and reduced
model.

\subsection{Example of application: cut set verification}

The above application to simple reachability does not requires the preservation of \emph{all} the
minimal traces.
Here, we apply the goal-oriented reduction to the cut sets for reachability, where the
\emph{completeness of minimal traces is crucial}.

Given a goal, a \emph{cut set} is a set of local states such that any trace leading to the
goal involves, in some of its transitions, one of these local states.
Therefore, disabling all the local states of a cut set should make the reachability of the goal
impossible.
This disabling could be implemented by the knock-out/in of the corresponding species in the
biological system:
cut sets predict mutations which should prevent a concerned reachability to occur (e.g., active
transcription factor).
Such cut sets have been studied in \cite{Sahin09,PAK13-CAV}  and are close to intervention sets
\cite{Klamt06} (which are not defined on traces but on pseudo-steady states).

We focus here on verifying if a (predicted) set of local states is, indeed, a cut set for the goal
reachability.
In the scope of this experiment, we consider cut sets that are disjoint with the initial state.
The cut set property can be expressed with CTL:
$\{a_1,b_1\}$ is a cut set for $g_\top$ reachability if the model satisfies the CTL property
\texttt{not E [ (not $a_1$ and not $b_1$) U $g_\top$ ]} (\texttt{U} being the \emph{until}
operator).
The property states that there exists no trace where none of the local state of the cut set is
reached prior to the goal.
It is therefore required that \emph{all} the minimal traces to the goal reachability are present in the
model: if one is missing, a set of local states could be validated as cut set whereas it may not be
involved in the missed trace.

\Tabref{cutsets} compares the model-checking of cut sets properties using \texttt{NuSMV} and
\texttt{its-ctl} \cite{itstools} on a range of the biological networks used in the previous
sections.
Because the dynamical property is much more complex, \texttt{its-reach} cannot be used.
The cut sets have been computed beforehand with \Pint{}.
Because the goal-oriented reduction preserves all the minimal traces to the goal, the results are
equivalent in the reduced models.
Similarly to the simple reachability, the goal-oriented reduction drastically improves the tractability of
large models.

\section{Discussion}
\label{sec:discussion}

This paper introduces a new reduction for automata networks parametrized by a
reachability property of the form: from a state $\sinit$ there exists a trace which
leads to a state where a given automaton $g$ is in state $g_\top$.

The goal-oriented reduction preserves \emph{all} the minimal traces satisfying
the reachability property under a general concurrent semantics which allows at
each step simultaneous transitions of an arbitrary number of automata.
Those results straightforwardly apply to the asynchronous semantics where only
one transition occurs at a time:
any minimal trace of the asynchronous semantics is a minimal trace in the
general concurrent semantics.

Its time complexity is polynomial in the total number of transitions and
exponential with the maximal number of local states within an automaton.
Therefore, the procedure is extremely scalable when applied on
networks between numerous automata, but where each automaton has a few local states.

Applied to logical models of biological networks, the goal-oriented reduction
can lead to a drastic shrinkage of the reachable state space with a negligible
computational cost.
We illustrated its application for the model-checking of simple reachability
properties, but also for the validation of cut sets, which requires the
completeness of minimal traces in the reduced model.
It results that the goal-oriented reduction can increase considerably the scalability of
the formal analysis of dynamics of automata networks.

The goal is expressed as a single local state reachability, which also allows to
to support sequential reachability properties between (sub)states
using an extra automaton.
For instance, the property ``reach $a_1$ and $b_1$, then reach $c_1$'' can be encoded
using one extra automaton $g$, where $\antrl g01{\{a_1,b_1\}}$ and $\antrl g1\top{\{c_1\}}$.

Further work consider performing the reduction on the fly, during the state
space exploration, expecting a stronger pruning.
Although the complexity of the reduction is low, such approaches
would benefit from heuristics to indicate when a new reduction step may be
worth to apply.

\bibliographystyle{plainurl}
\bibliography{bibliography,pauleve}
\appendix
\section{Proof of minimal traces preservation}
\suppllabel{gored-proof}

We assume a global AN $(\anN,\anS,\anT)$ where
$g\in\anN$, $g_\top\in\anS(g)$, and $\sinit\in\anS$ with
$\get{\sinit}{g}\neq g_\top$.

From \ptyref{csol} and \defref{valid}, any trace reaching first $a_i$ and then $a_j$ uses all
the transitions of at least one local path in $\rcsol\sinit(\anobj aij)$.

We first prove with \lemref{last-is-known}
that the last transition of a minimal trace $\pi$ for $g_\top$ reachability, of the form $\pi^{\card\pi}=\{\antr gi\top\}$, is necessarily in $\TRRed$.
Indeed, by definition of $\RLCG$, $\anobj g0\top\in\RLCG$;
and by \lemref{csol-end}, $\antr gi\top \notin
\rcsol\sinit(\anobj g0\top)$ implies that reaching $g_i$ requires to reach $g_\top$ beforehand.

\begin{lemma}
Given $\antr aji\in\anT$, if $\antr aji \notin \tr(\rcsol\sinit(\anobj a0i))$,
then for any trace $\pi$ from $\sinit$
with $a_j\in\postcond{\trace^v}$
and $a_i\in\postcond{\trace^w}$ for some $v,w\in\indexes\pi$,
there exists $u<v$ with $a_i\in\postcond{\trace^u}$.
\label{lem:csol-end}
\end{lemma}
\begin{proof}
Let $\lpath\in\rcsol\sinit(\anobj a0j)$ be an acyclic local path such that
$\forall n\in\indexes\lpath$, $a_i\neq\dest(\lpath^n)$.
The sequence $\lpath\concat \antr aji$ is then acyclic and, by definition, belongs to
$\rcsol\sinit(\anobj a0i)$, which is a contradiction.
\qed
\end{proof}

\begin{lemma}
If $\pi$ is a minimal trace for $g_\top$ reachability from state $\sinit$,
then, necessarily,
$\pi^{\card\pi}\subseteq\TRRed$.
\label{lem:last-is-known}
\end{lemma}
\begin{proof}
As $\pi$ is minimal for $g_\top$ reachability, without loss of generality, we can assume that
$\pi^{\card\pi} = \{\antr gi\top\}$.
By definition, $\tr(\rcsol\sinit(\anobj g0\top))\subseteq\TRRed$.
By \lemref{csol-end}, if $\antr gi\top\notin\tr(\rcsol\sinit(\anobj g0\top))$,
then there exists $u < \card\pi$ such that
$g_\top \in \postcond{\pi^u}$; hence, $\pi$ would be non minimal.
\qed
\end{proof}


The rest of the proof of \thmref{gored} is derived by contradiction:
if a transition of $\pi$ is not in $\TRRed$, we can build a sub-trace of $\pi$
which preserves $g_\top$ reachability, therefore $\pi$ is not minimal.

Given a transition $\antr aij$ in the $q$-th step of $\pi$ that is not in $\TRRed$, removing
$\antr aij$ from $\pi^q$ would imply to remove any further transition that depend causally
on it.
Two cases arise from this fact:
either all further transitions that depend on $a_j$ must be removed;
or $\antr aij$ is part of loop within automaton $a$, and it is sufficient to remove the loop from $\pi$.

%

\Lemref{cycle} ensures that if $\anobj azk$ is in $\RLCG$ and if
$a_z$ occurs before the $q$-th step and $a_k$ after the $q$-th step of $\trace$,
then $\antr aij\notin\tr(\rcsol{\sinit}(\anobj azk))$ only if
$\antr aij$ is part of a loop, i.e., there are two steps surrounding $q$
where the automaton $a$ is in the same state before their application.

\begin{lemma}
Given $a\in\anN$ and $u,q,v\in\indexes\trace$,
$u\leq q< v$,
with
$a_z \in\precond{\trace^u}$,
$a_k \in \precond{\trace^v}\cup\postcond{\trace^v}$,
and
$\antr aij\in\trace^q\setminus\TRRed$,
if $\anobj azk\in\RLCG$
then
$\exists m,n\in\range uv$, $m\leq q\leq n$ such that
$\postcond{(\trace^{1..m-1})}\cap\anS(a)
=\postcond{(\trace^{1..n})}\cap\anS(a)$;
and
$a_k\in\precond{\trace^v}\Rightarrow n <v$.
\lemlabel{cycle}
\end{lemma}
\begin{proof}
If $\antr aij\notin\TRRed$ and $\anobj azk\in\TRRed$, necessarily
$\antr aij\notin\tr(\rcsol\sinit(\anobj azk))$.
Therefore $\antr aij$ belongs to a loop of a local path from $a_z$ (at index $u$ in $\trace$) to $a_k$
(at index $v$ in $\trace$).
Hence, $\exists m,n\in\range{u}{v}$ with $m \leq q \leq n$
and $a_h,a_x,a_y\in \anS(a)$
such that
$\antr ahx\in\trace^m$ and
$\antr ayh\in\trace^n$;
therefore
$\postcond{(\trace^{1..m-1})} \cap \anS(a)
=\postcond{(\trace^{1..n})}\cap\anS(a) = a_h$.
In the case where $a_k\in\precond{\trace^v}$,
$a_k\neq a_h$, hence $n < v$.
\qed
\end{proof}

Intuitively, \lemref{cycle} imposes that $\pi$ has the following form:
\[
\begin{array}{ccccccccccc}
& \text{\rotatebox{270}{$a_z\in$}}
&&&& \text{\rotatebox[origin=c]{90}{$\notin\TRRed$}}
&&&& \text{\rotatebox{270}{$a_k\in$}}
\\
\pi = \cdots & \concat\pi^u\concat & \cdots & \concat \antr ahx \concat & \cdots &
\concat \mathbf{\antr aij} \concat & \cdots & \concat\antr ayh\concat & \cdots & \concat\pi^v\concat & \cdots
\\
&\scriptstyle u & & \scriptstyle m&& \scriptstyle q && \scriptstyle n && \scriptstyle v
\end{array}
\]
given that $\anobj azk\in\RLCG$.

The idea is then to remove the transitions forming the loop within automaton $a$.
However, transitions in other automata may depend causally on the transitions
that compose the local loop in automaton $a$ within steps $m$ and $n$, following the notations in \lemref{cycle}.

\Lemref{cb} establishes that we can always find $m$ and $n$ such that
none of the transitions within these steps with an enabling condition depending on automaton $a$
are in $\TRRed$.
Indeed, if a transition in $\TRRed$ depends on a local state of $a$, let us call it $a_p$, 
the objectives $\anobj a0p$ and $\anobj apk$ are in $\RLCG$, due to the second and third condition in
\defref{rlcg}.
\Lemref{cycle} can then be applied on the subpart of $\pi$ that contains the transition $\antr aij$
not in $\TRRed$ and that concretizes either $\anobj a0p$ or $\anobj apk$ to identify a smaller loop
containing $\antr aij$.

\begin{lemma}
Let us assume $a\in\anN$
and
$q\in\indexes\trace$
with
$\antr aij\in\trace^q\setminus\TRRed$.
There exists $m,n\in\indexes\trace$
with
$m\leq q\leq n$
such that
$\forall t\in\tr(\trace^{m+1..n})$,
$\dep(t)\cap\anS(a)\neq\emptyset\Rightarrow t\notin\TRRed$,
and,
if $a=g$ or
$\exists t\in\tr(\trace^{n+1..\card\trace})\cap\TRRed$
with
$\dep(t)\cap\anS(a)\neq\emptyset$,
then
$\postcond{(\trace^{1..m-1})}\cap\anS(a)
=\postcond{(\trace^{1..n})}\cap\anS(a)$
.
\lemlabel{cb}
\end{lemma}
\begin{proof}
First, let us assume that $a\neq g$ and
for any $t\in\trace^{q+1..\card\trace}$,
$\dep(t)\cap\anS(a)\neq\emptyset \Rightarrow t\notin \TRRed$:
the lemma is verified with $m=q$ and $n=\card\trace$.

Then, let us assume there exists $v\in\range{q+1}{\card\trace}$ such that
$\exists t\in\tr(\trace^v)\cap\TRRed$ with
$a_k\in\dep(t)$.
By \defref{rlcg},
this implies $\anobj a0k\in\RLCG$.
By \lemref{cycle}, there exists
$m,n\in\range 1 {v-1}$ with $m\leq q\leq n$ such that
$\postcond{(\trace^{1..m-1})}\cap\anS(a)
=\postcond{(\trace^{1..n})}\cap\anS(a)$.

Otherwise, $a=g$, and by \lemref{cycle} with $a_k=g_\top$, there exists
$m,n\in\range 1 {\card\trace}$ with $m\leq q\leq n$ and $m\neq n$ such that
$\postcond{(\trace^{1..m-1})}\cap\anS(a)
=\postcond{(\trace^{1..n})}\cap\anS(a)$.
Remark that it is necessary that $n < \card\trace$:
if $n=\card\trace$, $g_\top\in\postcond{(\trace^{1..m-1})}$, so $\trace$ would be not minimal.

In both cases, if there exists $r\in\range {m+1} n$ such that
$\exists a_p\in\anS(a)$ and $\exists t\in\trace^r$ with
$a_p\in\dep(t)$,
then $t\in\TRRed$ implies that $\anobj a0p\in\RLCG$ and $\anobj apk\in\RLCG$
(\defref{rlcg}).
If $r > q$, by \lemref{cycle} with
$a_k=a_p$ and $v=r$,
there exists $m',n'\in\range {m+1} n$
such that $m' \leq q \leq n' < r\leq n$ with
$\postcond{(\trace^{1..m'-1})}\cap\anS(a)
=\postcond{(\trace^{1..n'})}\cap\anS(a)$.
If $r \leq q$, by \lemref{cycle} with
$a_0=a_p$ and $u=r$,
there exists $m',n'\in\range {m+1}n$
such that $r \leq m' \leq q \leq n'$ with
$\postcond{(\trace^{1..m'-1})}\cap\anS(a)
=\postcond{(\trace^{1..n'})}\cap\anS(a)$.
Therefore, by induction with \lemref{cycle},
there exists $m,n\in\indexes\trace$
such that
$\forall t\in\tr(\trace^{m+1..n})$,
$\dep(t)\cap\anS(a)\neq\emptyset \Rightarrow t\notin\TRRed$.
\qed
\end{proof}

Using \lemref{cb}, we show how we can identify a subset of transitions in $\trace$ that can be
removed to obtain a sub-trace for $g_\top$ reachability.
In the following, we refer to the couple $(m,n)$ of \lemref{cb}
with $\cb(\pi,a,q)$ (\defref{cb}).

\begin{definition}[$\cb(\trace,a,q)$]
\deflabel{cb}
Given $a\in\anN$, $q\in\indexes\trace$ with $t\in\trace^q\setminus\TRRed$ and $\anN(t)=a$,
we define $\cb(\pi,a,q) = (m,n)$ where $m,n\in\indexes\trace$ such that:
\begin{itemize}
\item $\forall t\in\tr(\trace^{m+1..n})$, $\dep(t)\cap\anS(a)\neq\emptyset\Rightarrow 
t\notin\TRRed$;
\item 
$a=g\vee \exists t\in\tr(\trace^{n+1..\card\trace})\cap\TRRed$
with
$\dep(t)\cap\anS(a)\neq\emptyset$
$\Longrightarrow$
$\postcond{(\trace^{1..m-1})}\cap\anS(a)
=\postcond{(\trace^{1..n})}\cap\anS(a)$.
Moreover, if $a=g$, then $n<\card\trace$.
\end{itemize}
\end{definition}


We use \lemref{cb} to collect the portions of $\pi$ to redact according to each automaton.
We start from the last transition in $\pi$ that is not in $\TRRed$:
if $\tr(\pi)\not\subseteq\TRRed$,
there exists $l \in\indexes\pi$ such that
$\pi^l\not\subseteq\TRRed$
and $\forall n>l, \pi^n\subseteq\TRRed$.
By \lemref{last-is-known}, we know that $l < \card\trace$.
Let us denote by $\antr bij$ one of the transitions in $\pi^l$ which is not in $\TRRed$.

We define
$\redact \subseteq \anN \times \indexes\trace \times \indexes\trace$
the smallest set which satisfies:
\begin{itemize}
\item $(b,m,n) \in\redact$ if $\cb(\trace,l,b) = (m,n)$
\item $\forall(a,m,n)\in\redact$,
	$\forall q\in[m+1;n]$,
	$\forall t\in\trace^q$,
 $\dep(t)\cap\anS(a)\neq\emptyset \Longrightarrow (\anN(t),m',n') \in\redact$
		where $\cb(\trace,q,\anN(t)) = (m',n')$.
\end{itemize}

Finally, let us define the sequence of steps $\traceb$
as the sequence of steps $\pi$ where the transitions delimited by $\redact$ are removed:
for each $(a,m,n)\in\redact$, all the transitions of automaton $a$ occurring
between $\trace^m$ and $\trace^n$ are removed.
Formally,
$\card\traceb=\card\trace$
and
for all $q\in\indexes\trace$,
$\traceb^q \DEF \{ t \in\trace^q\mid \nexists (a,m,n)\in\redact: a=\anN(t) \wedge
m\leq q\leq n \}$.

From \lemref{cb} and $\redact$ definition, $\traceb$ is a valid trace.
Moreover, by \lemref{cb}, there is no $q\in\indexes\trace$ such that
$(g,q,\card\trace)\in\redact$, hence $g_\top \in\postcond\traceb$.
Therefore, $\trace$ is not minimal, which contradicts our hypothesis.
$\qed$

\begin{example}
Let us consider the reachability of $c_2$ in the AN of \figref{example1} from state
$\state{a_0,b_0,c_0,d_0}$.
The transitions $\TRRed$ preserved by the reduction for that goal
are listed in \figref{example1r}.

Let $\pi$ be the following trace in the AN of \figref{example1}:
\begin{equation*}
\begin{split}
\pi&=\{\antrl a01{\{b_0\}}\}\concat
	\{\antrl b01{\{a_1\}},\antrl c01{\{a_1\}}\}\concat
	\{\antrl a10{\emptyset}\}\concat
	\{\antrl b10{\{a_0\}}\}
	\\&\qquad \concat
	\{ \antrl c12{\{b_0\}}\}\enspace.
\end{split}
\end{equation*}
The latest transition not in $\TRRed$ is $\antrl b10{\{a_0\}}$ at step $4$.
One can compute $\cb(\pi,4,b)=(2,4)$, and as there is no transition involving $b$ between steps $3$ and $4$,
$\redact = \{(b,2,4)\}$; therefore, the sequence
\begin{equation*}
\varpi=\{\antrl a01{\{b_0\}}\}\concat
	\{\antrl c01{\{a_1\}}\}\concat
	\{\antrl a10{\emptyset}\}\concat
	\{\}\concat
	\{ \antrl c12{\{b_0\}}\}
\end{equation*}
is a valid sub-trace of $\pi$ reaching $c_2$, proving $\pi$ non-minimality.
\end{example}

In conclusion, if $\trace$ is a minimal trace for $g_\top$ reachability from
state $\sinit$, then, $\tr(\trace)\subseteq\TRRed$.

\section{Experiments with partial reduction}
\suppllabel{partial}

The goal-oriented reduction relies on two intertwined analyses of the local causality in ANs:
(1) the computation of potentially involved objectives (\secref{gored-proc}) and 
(2) the filtering of objective that can be proven impossible (\secref{gored-filter}).
The second part can be considered optional: one could simply define the predicate
$\valid{\sinit}$ to be always true.
In order to appreciate the effect of this second part, we show here the intermediary results of
model reduction without the filtering of impossible objectives.
It is shown in table below, in the lines in \textit{italic}.
As we can see, for some models it has no effect on the reduction, for some others the filtering
parts is necessary to obtained important reduction of the state space (e.g., MAPK, TCell-r (94),
TCell-d).

\begin{scriptsize}

\begin{longtable}{|:l||^r|^r|^r|}
\hline
Model & \# tr & \# states & |unf| \\\hline
\multirow{3}{*}{EGF-r (20)}
& 68 & 4,200 & 1,749 \\
\rowstyle{\itshape} & 
	43 & 722 & 336\\
& \textbf{43} & \textbf{722} & \bfseries 336 \\\hline
\multirow{3}{*}{Wnt (32)}
	& 197 & 7,260,160 & KO\\
\rowstyle{\itshape} &
	134 &241,060 & 217,850\\
\rowstyle{\bf}
	& 117 & 241,060 & 217,850 \\\hline
\multirow{3}{*}{TCell-r (40)}
	& 90 & $\approx 1.2\cdot 10^{11}$ &KO\\
\rowstyle{\itshape}
	& 46 & 25,092 & 14,071 \\
\rowstyle{\bf}
	& 46 & 25,092 & 14,071 \\\hline
MAPK (53)
	& 173 & $\approx 3.8\cdot 10^{12}$ &KO\\
profile 1
\rowstyle{\itshape} &
	147 & $\approx 9\cdot 10^{10}$ & KO \\
 \rowstyle{\bf}
	& 113 & $\mathbf{\approx 4.5\cdot 10^{10}}$ & KO \\\hline
MAPK (53)
	& 173 & 8,126,465 & KO \\
profile 2 \rowstyle{\itshape}
	& 148 & 1,523,713 & KO \\
\rowstyle{\bf}
	& 69 & 269,825 & 155,327 \\\hline
\multirow{3}{*}{VPC (88)}
	& 332 & KO & KO \\
\rowstyle{\itshape} 
	& 278 & $\approx 2.9\cdot 10^{12}$ & 185,006\\
\rowstyle{\bf}
	& 219 & $\mathbf{1.8\cdot 10^9}$ & 43,302 \\\hline
\multirow{3}{*}{TCell-r (94)}
 &	217& KO  & KO \\
 \rowstyle{\itshape}
 	& 112 & KO & KO\\
& \textbf{42} & \textbf{54.921} & \bf 1,017 \\\hline
TCell-d (101) &
	384 & $\approx 2.7\cdot 10^8$ & 257 \\
profile 1 
\rowstyle{\itshape} &
	275 & $\approx 1.1\cdot 10^8$ & 159\\
\rowstyle{\bf} &
	0 & 1 & 1 \\\hline
TCell-d (101) &
	384 & KO & KO \\
profile 2
\rowstyle{\itshape}& 
	253 & $\approx 2.4\cdot 10^{12}$ & KO \\
\rowstyle{\bf} &
	161 & 75,947,684 & KO \\\hline
EGF-r (104) & 
	378 & 9,437,184 & 47,425 \\
profile 1
\rowstyle{\itshape} &
	120 & 12,288 & 1,711\\
	&
	\textbf{0} & \textbf{1} & \textbf{1} \\\hline
EGF-r (104) &
	378 & $\approx 2.7 \cdot 10^{16}$ & KO\\
profile 2
\rowstyle{\itshape} &
	124 & $\approx 2\cdot 10^9$ & KO\\
\rowstyle{\bf} & 
	69 & 62,914,560 & KO\\\hline
\multirow{2}{*}{RBE2F (370)} &
	742 & KO & KO \\
\rowstyle{\itshape}&
	56 & 2,350,494 & 28,856 \\
\rowstyle{\bf} &
	56 & 2,350,494 & 28,856 \\\hline
\end{longtable}
\end{scriptsize}

\end{document}